\documentclass[12pt]{article}
\RequirePackage[colorlinks,citecolor=blue,urlcolor=blue,linkcolor=blue]{hyperref}
\usepackage{url}

\usepackage[ansinew]{inputenc}

\usepackage{doi}

\usepackage{amsmath}
\usepackage{bbm}

\usepackage{mathtools}

\usepackage{amsfonts}
\usepackage{amsthm}
\usepackage{amssymb}

\usepackage{color}
\usepackage{bm} 
\usepackage{dsfont} 

\usepackage{graphicx}
\usepackage[lmargin=1.5cm,rmargin=1.5cm,bottom=3cm,top=3cm]{geometry}

\usepackage{tikz}
\definecolor{mycolor1}{rgb}{0.105882,0.619608,0.466667}
\definecolor{mycolor2}{rgb}{0.85098,0.372549,0.00784314}
\definecolor{mycolor3}{rgb}{0.458824,0.439216,0.701961}
\definecolor{mycolor4}{rgb}{0.905882,0.160784,0.541176}
\definecolor{mycolor5}{rgb}{0.4,0.65098,0.117647}
\definecolor{mycolor6}{rgb}{0.65098,0.462745,0.113725}
\definecolor{mycolor7}{rgb}{0.901961,0.670588,0.00784314}
\definecolor{mycolor8}{rgb}{0.4,0.4,0.4}
\definecolor{mycolor9}{rgb}{0.301961,0,0.294118}
\definecolor{mycolor10}{rgb}{0.0313725,0.25098,0.505882}

\usetikzlibrary{calc}        

\makeatletter
\newif\ifmygrid@coordinates
\tikzset{/mygrid/step line/.style={line width=0.80pt,draw=gray!80},
         /mygrid/steplet line/.style={line width=0.25pt,draw=gray!80}}
\pgfkeys{/mygrid/.cd,
         step/.store in=\mygrid@step,
         steplet/.store in=\mygrid@steplet,
         coordinates/.is if=mygrid@coordinates}
\def\mygrid@def@coordinates(#1,#2)(#3,#4){%
    \def\mygrid@xlo{#1}%
    \def\mygrid@xhi{#3}%
    \def\mygrid@ylo{#2}%
    \def\mygrid@yhi{#4}%
}
\newcommand\DrawGrid[3][]{%
    \pgfkeys{/mygrid/.cd,coordinates=true,step=1,steplet=0.2,#1}%
    \draw[/mygrid/steplet line] #2 grid[step=\mygrid@steplet] #3;
    \draw[/mygrid/step line] #2 grid[step=\mygrid@step] #3;
    \mygrid@def@coordinates#2#3%
    \ifmygrid@coordinates%
        \draw[/mygrid/step line]
        \foreach \xpos in {\mygrid@xlo,...,\mygrid@xhi} {%
          (\xpos,\mygrid@ylo) -- ++(0,-3pt)
                              node[anchor=north] {$\xpos$}
        }
        \foreach \ypos in {\mygrid@ylo,...,\mygrid@yhi} {%
          (\mygrid@xlo,\ypos) -- ++(-3pt,0)
                              node[anchor=east] {$\ypos$}
        };
    \fi%
}
\makeatother

\usepackage{psfrag}

\usepackage{algorithm}
\usepackage{algpseudocode}
\algdef{SE}[DOWHILE]{Do}{doWhile}{\algorithmicdo}[1]{\algorithmicwhile\ #1}%

\usepackage{mathrsfs} 

\usepackage[square,numbers,compress]{natbib}

\newcommand{\remove}[1]{}
\newcommand{\removesafe}[1]{}

\usepackage{subfigure}

\newcommand{\argmax}[1]{\underset{#1}{\operatorname{argmax}}}

\newcommand{\1}{\mathbbm{1}}

\newcommand{\floor}[1]{\lfloor #1 \rfloor}

\DeclareMathOperator*{\sign}{sgn}
\DeclareMathOperator*{\supp}{supp}
\newcommand{\xo}{x_0}

\newcommand{\R}{\mathbb{R}}

\newcommand{\etap}{\eta^+}
\newcommand{\etam}{\eta^-}
\newcommand{\Omegap}{\Omega^+}
\newcommand{\Omegam}{\Omega^-}
\newcommand{\Omegac}{\Omega^c}

\newcommand{\eps}{\varepsilon}
\newcommand{\inddef}{\mathbbm{1}_{|\langle a_i, \xo \rangle| \leq 3}}
\newcommand{\inddefwithnorm}{\mathbbm{1}_{|\langle a_i, \xo \rangle| \leq 3 \| \xo \|_2}}
\newcommand{\ind}{\mathbbm{1}_i}
\newcommand{\aixo}{\langle a_i, \xo \rangle}
\newcommand{\aixone}{\langle a_i, x_1 \rangle}
\newcommand{\aih}{\langle a_i, h \rangle}

\newcommand{\xoxohh}{\|\xo\|_2 \|h\|_2}
\newcommand{\xtilde}{\tilde{x}}
\newcommand{\etilde}{\tilde{e}}
\newcommand{\A}{\mathcal{A}}
\newcommand{\Axo}{\A_{\xo}}
\newcommand{\PP}{\mathbb{P}}
\newcommand{\aiai}{a_i a_i^\intercal}
\newcommand{\zetaizetai}{\zeta_i \zeta_i^\intercal}
\newcommand{\sigmin}{\sigma_\text{min}}
\newcommand{\sigmax}{\sigma_\text{max}}

\newtheorem{theorem}{Theorem} 
\newtheorem{lemma}[theorem]{Lemma}
\newtheorem{proposition}[theorem]{Proposition}

\title{Corruption Robust Phase Retrieval via Linear Programming }
\author{
Paul Hand\footnote{Department of Computational and Applied Mathematics, Rice University, TX.} \   and Vladislav Voroninski\footnote{Helm.ai, CA}}

\begin{document}

\maketitle
\abstract{
We consider the problem of phase retrieval from corrupted magnitude observations. In particular we show that a fixed $x_0 \in \mathbb{R}^n$ can be recovered \emph{exactly} from corrupted magnitude measurements $|\langle a_i, x_0 \rangle | + \eta_i, \quad i =1,2\ldots m$ with high probability for $m = O(n)$, where  $a_i \in \mathbb{R}^n$ are i.i.d standard Gaussian and $\eta \in \mathbb{R}^m$ has fixed sparse support and is otherwise arbitrary, by using a version of the PhaseMax algorithm augmented with slack variables subject to a penalty.  This linear programming formulation, which we call RobustPhaseMax, operates in the natural parameter space, and our proofs rely on a direct analysis of the optimality conditions using concentration inequalities.
}

\section{Introduction}
Recovering signals from corrupted data is a prevalent problem in modern applied science. For example, data corruption may be due to sensor failure, adversarial tampering as in cybersecurity, or gross errors that are inherent to necessary preprocessing as in the structure from motion problem in computer vision. Appropriately handling such corrupted data is a challenging aspect of algorithm design, and falls broadly speaking under the umbrella of robust statistics \cite{huber2011robust}. A prominent example is that of principal component analysis, where standard solutions based on singular value decompositions are optimal under Gaussian noise but can be arbitrarily skewed by a single corrupted data point. Much work has focused on formulating robust versions of PCA, culminating in breakthrough results by Candes et al. \cite{RPCA}, which exploit concepts from compressed sensing and matrix completion to formulate a convex program that operates in the space of matrices to decouple low rank and sparse structures, the latter of which model corruptions. In addition to its favorable theoretical properties, the robust PCA framework has found much empirical success in a wide variety of applications.  There has also been work on corruption and erasure robustness in other linear settings such as compressed sensing and matrix completion \cite{Li2013, voroninski2016strong, CJSC2013}.

Beyond the linear model setting, there are many situations in which signal recovery must be performed from corrupted nonlinear observations. One such scenario is phase retrieval, in which corrupted data is a natural byproduct of some of its most promising applications, such as Xray Free Electron Laser high energy imaging.  In this imaging scenario a stream of identical particles is blasted by very high intensity lasers, yielding a multitude of highly noisy diffraction images of random rotations of the particle \cite{xfel1, xfel2}. Thus, the task at hand is recovering information from redundant yet highly corrupted data.  We introduce here a linear program for phase retrieval in the natural parameter space called RobustPhaseMax, which we show recovers real vectors exactly with high probability from $O(n)$ random Gaussian magnitude measurements, $\Omega(n)$ of which are arbitrarily corrupted. The RobustPhaseMax formulation is highly amenable to analysis, and our proofs rely on a simple characterization of its optimality conditions and standard concentration inequalities.

The first theoretical guarantees for stable and efficient phase retrieval at optimal sample complexity were for the PhaseLift algorithm, which is a convex program in a lifted space of matrices \cite{CL2012, CSV2013}. Subsequently, an L1-aware version of PhaseLift was shown to be robust to noise in \cite{CL2012}, and to an $O(1)$ fraction of arbitrary outliers, in addition to noise, in \cite{phaseliftrobust}. While having favorable theoretical properties, the PhaseLift framework is not computationally tractable in high dimensional regimes due to the squaring of the natural dimensionality of the problem at hand. 

Meanwhile, lifting to higher dimensional spaces is not necessary for successful algorithm design for non-linear inverse problems --- even for those that require  robustness to outliers. For instance, nonlinear corruption-robust recovery in the natural parameter space was recently achieved in computer vision by the empirical success of the LUD algorithm for location recovery from corrupted relative directions \cite{Amit}. It was shown by \cite{ShapeFit} that an alternative convex program called ShapeFit, recovers locations exactly from corrupted relative directions under broad conditions on the locations and the graph of relative direction observations. This location recovery problem can be interpreted as a robust PCA problem under projective ambiguity in each observation. That is, if we observed the pairwise distances as well as the relative directions, the problem would be an instance of structured rank-2 robust PCA problem. Remarkably, instead of lifting to a space of matrices, ShapeFit operates in the natural parameter space and still achieves the goals of robust PCA, even while subject to projective ambiguity in each observation and adversarial corruptions. 

Indeed, corruption-robust recovery has also been achieved for phase retrieval in the natural parameter space in \cite{medianwirtingerflow}, by modifying the Wirtinger Flow framework of Candes et al. \cite{wirtinger,cai2015optimal}. Unlike ShapeFit, which consists of a single convex program, Wirtinger Flow approaches rely on a spectral initialization, followed by a non-convex optimization program. 

Recently, a linear programming formulation for phase retrieval, called PhaseMax, was independently discovered by \cite{phasemax, phasemaxJustin}. When provided with an anchor vector that correlates with the desired signal,  PhaseMax has been shown to recover signals from magnitude measurements with high probability.  Several recovery theorems exist for PhaseMax, some involving tools from statistical learning theory \cite{phasemaxJustin}, geometric probability \cite{phasemax}, and elementary probabilistic concentration arguments \cite{phasemaxelementary}.  A sparsity-aware version called SparsePhaseMax recovers sparse signals at optimal sample complexity when properly initialized \cite{sparsephasemax}.  
While noise stability for PhaseMax has been established \cite{phasemax, phasemaxJustin}, the assumptions on the noise are in the infinity norm, which are not realistic for Gaussian or even Poisson noise. At first glance, it isn't clear whether the PhaseMax framework can accommodate more realistic noise models due to the structure of its feasible set. Namely, the feasible set of PhaseMax can shrink very rapidly into a subset of a small ball around the origin when $O(n)$ measurements  significantly underestimate their true value, as may occur in the regime of highly redundant yet corrupted data. In this paper, we introduce a outlier-tolerant algorithm called RobustPhaseMax.  It consists of augmenting the PhaseMax linear program with non-negative slack variables that are subject to a penalty on their sum.  

One way to get an anchor vector for RobustPhaseMax is through the 
initializer provided by the Median-Truncated Wirtinger Flow work of \cite{medianwirtingerflow}, which succeeds with high probability under $m = O(n)$ Gaussian measurements and a constant fraction of corruptions.   Thus, our results, together with the initializer of \cite{medianwirtingerflow}, guarantee exact phase retrieval from $O(n)$ corrupted measurements. This is the first result of its kind for corruption-robust phase retrieval in two respects. Firstly, RobustPhaseMax is a convex program in the natural parameter space. Secondly, the best previously published result that operates in the natural parameter space is that of \cite{medianwirtingerflow}, which relies on $O(n \log n)$ measurements due to the post-initialization step of the Median-Truncated Wirtinger Flow pipeline. In sum, our work establishes that despite the non-linear nature of magnitude measurements, exact phase retrieval can be performed robustly to corruptions in the natural parameter space via linear programming, at optimal sample complexity.

\section{Main result}

Let $\xo \in \R^n$.  Let $a_i \sim \mathcal{N}(0, I_{n \times n})$ for $i=1 \ldots m$.  Let $b_i = | \langle a_i, \xo \rangle | + \eta_i$, where $\eta\in\R^m$ is a sparse vector of  corruptions of arbitrary magnitude and sign.  If an anchor vector $\phi$ is known to approximate the signal $\xo$, then we consider recovering $\xo$ by the following linear program, called RobustPhaseMax:
\begin{equation}\label{RPM}
\begin{array}{ll}
\max & \langle \phi, x \rangle - \lambda \langle 1, e \rangle \\[.5em]
\text{s.t.} & -b_i - e_i \leq \langle a_i, x \rangle \leq b_i + e_i, \quad e_i \geq 0, \quad   i = 1\ldots  m, \quad x \in \mathbb{R}^n, \ e \in \R^m.\\ 
\end{array}
\end{equation}
The variables $\{e_i\}_{i=1}^m$ are slack variables, $1$ represents the $m$-vector of all ones, and $\lambda > 0$.  We show that this linear program recovers $\xo$ exactly with high probability from $O(n)$ Gaussian measurements, of which an adversarially chosen $O(1)$-fraction are subject to arbitrary corruptions, provided that $\phi$ is sufficiently accurate and $\lambda = \Theta(\|\xo\|_2/m)$.

\begin{theorem} \label{thm:rpm}
For any $M\geq 7$, there exist  positive  $c, \gamma, \delta$ such that the following holds.  Let $\xo \in \R^n$.  Let $\phi \in \R^n$ be such that $\|\phi - \xo\|_2 < 0.5 \|\xo\|_2$.    Let $a_i \sim \mathcal{N}(0, I_{n \times n})$ be independent for $i = 1 \ldots m$.   If  $m \geq c n$ and $\lambda \in \left[ \frac{7}{m} \|\xo\|_2, \frac{M}{m} \|\xo\|_2\right]$, then with probability at least $1 - 9 m e^{-\gamma m}$, it holds that, simultaneously  for  all $\eta \in \R^m$ such that $| \supp (\eta) | \leq \delta m$, $(\xo, -\etam)$ is the unique maximizer of \eqref{RPM} for $b_i = | \aixo| + \eta_i$, $i = 1\ldots m$.  Here, $\etam  = \min(\eta, 0)$.

\end{theorem}

We note that when recovery occurs under the conditions of this theorem, the slack variables $e$ act to correct only the measurements that are below their true value.  Finally, RobustPhaseMax is equivalent to another natural formulation for corruption-robust phase retrieval as follows.

%

\begin{proposition} \label{prop:equiv}
Let $\phi \in \R^n$, $\lambda \in \R$, $\lambda>0$, $a_i \in \R^n$ for $i = 1 \ldots m$, $b \in \R^m$.  Then  the following are equivalent:
\begin{alignat}{3}
(\xtilde, \etilde) \in  &\argmax{x \in \R^n, \ e \in \R^m} \langle \phi, x \rangle - \lambda \| e\|_1 &&\text{subject to } | \langle a_i, x \rangle | \leq b_i + e_i \label{rpm-l1}, \quad &&i = 1\ldots m,\\
(\xtilde, \etilde) \in &\argmax{x \in \R^n, \ e \in \R^m} \langle \phi, x \rangle - \lambda \langle 1, e \rangle \ &&\text{subject to } | \langle a_i, x \rangle | \leq b_i + e_i, e_i \geq 0, \quad &&i = 1\ldots m. \label{rpm-pos}
\end{alignat}
\end{proposition}
This proposition follows from the fact that if $(\xtilde, \etilde)$ is feasible for the program in \eqref{rpm-l1}, then $(\xtilde, \max(0,\etilde))$ is also feasible and has a strictly greater objective if $\etilde$ has any negative entries. Thus, any $(\xtilde, \etilde)$ that satisfies \eqref{rpm-l1} is such that $\etilde \geq 0$.


\section{Proofs}
\begin{proof}[Proof of Theorem \ref{thm:rpm}]
Let $\Omegap = \{ i \mid \eta_i > 0\}$, $\Omegam = \{ i \mid \eta_i < 0\}$, $\Omega = \Omegap \cup \Omegam$,  and $\eta = \etap + \etam$ with $\etap = \max(\eta, 0)$ and $\etam = \min(\eta, 0)$.
Consider a feasible point $(\xo + h, -\etam + g)$.  To show $(\xo, -\etam)$ is the unique solution, it suffices to show that for any feasible perturbation $(h,g) \neq (0, 0)$, 
\begin{align}
\langle \phi, h \rangle - \lambda \langle 1, g \rangle < 0. \label{objective-descent}
\end{align}
Because of the feasibility of the perturbation $(h,g)$, 
\[
\sign ( \langle a_i, \xo \rangle ) \langle a_i, \xo + h \rangle \leq | \aixo | + \etap_i + \etam_i - \etam_i + g_i =  | \langle a_i, \xo \rangle | + \etap_i + g_i,
\]
which implies 
\begin{alignat}{2}
&\sign (\langle a_i, \xo \rangle) \langle a_i, h \rangle \leq g_i \quad &&\text{for all } i \in \Omegac \cup \Omegam. \label{feasibility-condition}
\intertext{
Define $S := \{ i \mid \langle a_i, \xo \rangle \langle a_i, h \rangle > 0 \}$.   We will need the following special case of \eqref{feasibility-condition}:
}
&g_i > 0 \text{ and } | \langle a_i, h \rangle | \leq g_i \quad &&\text{for all } i \in  S \cap (\Omegac \cup \Omegam). \label{feas-cond-S}
\end{alignat}
By the feasibility condition $e_i \geq 0, i =1\ldots m$, we also have 
\begin{align}
g_i \geq 0 \text{ for } i \in \Omegac \cup \Omegap. \label{gpositive}
\end{align}

We begin by establishing \eqref{objective-descent}  if $h=0$ and $g \neq 0$.  In this case, note that 
$-\lambda \langle 1, g \rangle  = -\lambda \| g\|_1 < 0$,
where the  equality follows because of \eqref{gpositive} and the combination of \eqref{feasibility-condition} with $h=0$. 

In the remainder of this proof, we assume $h \neq 0$ and establish \eqref{objective-descent}.  Write
\begin{align}
\langle \phi, h\rangle - \lambda \langle 1, g \rangle &= \langle \xo, h \rangle - \lambda \sum_{i \in \Omegac} g_i - \lambda \sum_{i \in \Omegam} g_i - \lambda \sum_{i \in \Omegap} g_i + \langle \phi - \xo , h \rangle \notag\\
&\leq \underbrace{\langle \xo, h \rangle - \lambda \sum_{i \in \Omegac} g_i}_{I} \underbrace{- \lambda \sum_{i \in \Omegam} g_i}_{II}  + \langle \phi - \xo , h \rangle, 
\label{terms-to-bound}
\end{align}
where the inequality follows by \eqref{gpositive}.  

First, we bound term $I$.   Let $\ind = \inddefwithnorm$ be the indicator of the event on which  $| \aixo| \leq 3 \| \xo\|_2$.   By Lemma \ref{lem:isometry-allomega-withind}, there exist positive $\delta, c_0, \gamma$ such that if $|\Omega| \leq \delta m$ and if $m \geq 2 c_0 n$, 
\begin{align}
\langle \xo, h \rangle &= \left \langle \frac{1}{|\Omegac|} \sum_{i \in \Omegac} \ind \cdot a_i a_i^\intercal, h \xo^\intercal \right \rangle - \left \langle \frac{1}{|\Omegac|} \sum_{i\in \Omegac} \ind \cdot a_i a_i^\intercal - I_{n\times n} , h \xo^\intercal \right \rangle \notag \\
&\leq \frac{1}{|\Omegac|} \sum_{i \in \Omegac} \ind \cdot \langle a_i, \xo \rangle \langle a_i, h\rangle + 0.04 \| \xo \|_2 \|h\|_2, \label{isometry-bound}
\end{align}
on an event with probability at least $1 - 2 m e^{-\gamma m/4}$.   We write
\begin{align}
I &= \langle \xo, h \rangle - \lambda \sum_{i \in \Omegac} g_i \notag\\
& \leq \frac{1}{|\Omegac|} \sum_{i \in \Omegac} \ind \cdot \aixo \aih - \lambda \sum_{i \in \Omegac} g_i + 0.04 \xoxohh  \notag\\
& = \frac{1}{|\Omegac|} \sum_{i \in \Omegac \cap S^c} \ind \cdot \aixo \aih + \frac{1}{|\Omegac|} \sum_{i \in \Omegac \cap S} \ind \cdot \aixo \aih - \lambda \sum_{i \in \Omegac} g_i + 0.04 \xoxohh  \notag\\
& \leq \frac{1}{|\Omegac|} \sum_{i \in \Omegac \cap S^c} -\ind \cdot |\aixo||\aih|  + \frac{1}{|\Omegac|} \sum_{i \in \Omegac \cap S} \ind \cdot \aixo \aih - \lambda \sum_{i \in \Omegac \cap S} g_i + 0.04 \xoxohh  \notag\\
& \leq \frac{1}{|\Omegac|} \sum_{i \in \Omegac \cap S^c} -\ind \cdot |\aixo||\aih|  + \frac{1}{|\Omegac|} \sum_{i \in \Omegac \cap S} \ind \cdot |\aixo| g_i - \lambda \sum_{i \in \Omegac \cap S} g_i + 0.04 \xoxohh  \notag\\
& \leq \frac{1}{|\Omegac|} \sum_{i \in \Omegac \cap S^c} -\ind \cdot |\aixo||\aih|  - \frac{1}{|\Omegac|} \sum_{i \in \Omegac \cap S} \ind \cdot |\aixo| g_i + 0.04 \xoxohh  \notag\\
& \leq \frac{1}{|\Omegac|} \sum_{i \in \Omegac \cap S^c} -\ind \cdot |\aixo||\aih|  - \frac{1}{|\Omegac|} \sum_{i \in \Omegac \cap S} \ind \cdot |\aixo| |\aih| + 0.04 \xoxohh  \notag\\
&= -\frac{1}{|\Omegac|} \sum_{i \in \Omegac} \ind \cdot |\aixo||\aih| + 0.04 \xoxohh,  \label{aixoaihxoxohh}
\end{align}
where the second line follows from \eqref{isometry-bound};  the fourth line uses the fact that $\aixo \aih \leq 0$ on $S^c$ and the fact that $g_i \geq 0$ for $i \in \Omegac$; the fifth line uses the feasibility condition \eqref{feasibility-condition};  the sixth line uses the fact that $g_i \geq 0$ for $i \in \Omegac$, the  assumption that $\lambda \geq \frac{7}{m} \|\xo\|_2$, which implies $\lambda  \geq \frac{6}{|\Omegac|} \| \xo\|_2$ for $\delta < 1/2$, and the fact that $\ind \cdot | \aixo | \leq \ind \cdot 3 \| \xo \|_2$; and the seventh line uses the feasibility condition \eqref{feas-cond-S}. 
By Lemma \ref{lemma:lower-bound-ind-uniformomega}, if $m \geq 2 c_0 n$ for some $c_0$, then on an event with probability at least $1-5 m e^{-\gamma m/4}$, 
\begin{align*}
\frac{1}{|\Omegac|} \sum_{i\in\Omegac} \ind \cdot |\aixo| |\aih| \geq 0.55 \xoxohh.
\end{align*}
By \eqref{aixoaihxoxohh}, we thus have 
\begin{align}
I \leq -0.51 \xoxohh. \label{termi}
\end{align}

Second, we now bound term $II$.
\begin{align*}
II = -\lambda \sum_{i \in \Omegam} g_i  &\leq -\lambda \sum_{i \in \Omegam} \sign(\aixo) \aih \leq \lambda \sum_{i \in \Omegam} |\aih|,
\end{align*}
where the first inequality follows by \eqref{feasibility-condition}.
By Lemma \ref{lem:aiho-concentration-uniformomega}, if $m \geq \frac{n}{\delta}$ for some $\delta$, and if $|\Omegam| \leq \delta m$, then on an event with probability at least $1 - 2 e^{-\delta m}$, 
\begin{align*}
\sum_{i \in \Omegam} | \aih | \leq (\sqrt{4 + 2 \log(1/\delta)} + 2)  \delta  m \cdot \|h\|_2.
\end{align*}
Using $\lambda \leq \frac{M}{m} \| \xo \|_2$, and choosing $\delta$ small enough so that $M \left(\sqrt{4 + 2 \log(1/\delta)} + 2 \right) \delta \leq 0.01$, 
\begin{align}
II \leq \frac{M}{m}\|\xo\|_2 \cdot  \left(\sqrt{4 + 2 \log(1/\delta)} + 2 \right) \delta m \cdot \|h\|_2 \leq 0.01 \xoxohh. \label{termii}
\end{align}

Finally, we conclude from \eqref{terms-to-bound}, \eqref{termi}, \eqref{termii}, and the assumption that $\| \phi - \xo \| < 0.5 \| \xo \|_2$ that
\begin{align*}
\langle \phi, h \rangle - \lambda \langle 1, g \rangle &\leq -0.51 \xoxohh + 0.01 \xoxohh + \| \phi - \xo \|_2 \|h\|_2 < 0.
\end{align*}

\end{proof}

We now prove several technical lemmas used in the proof above.  The first two lemmas concern probabilistic concentration of the eigenvalues of appropriately truncated averages of $a_i a_i^\intercal$ for $i$ belonging to an arbitrarily selected $\Omegac$ of appropriate size.

\begin{lemma} \label{lem:aiai-concentration-ind}
There exist positive $c_0, \gamma$ such that the following holds.  Fix $\xo \in \R^n$.  Let $a_i \sim \mathcal{N}(0, I_{n \times n}), i = 1 \ldots m$ be independent.  If $m \geq c_0 n$, then with probability at least $1 - 2 e^{-\gamma m}$,
\begin{align*}
\left \| \frac{1}{m} \sum_{i=1}^m \inddefwithnorm \cdot \aiai - I_{n \times n}\right \| \leq 0.04.
\end{align*}
\end{lemma}
\begin{proof}
If $\xo = 0$, then the lemma follows directly from standard concentration estimates of Gaussian matrices, such as Corollary 5.35 in \cite{V2012}.  Consider $\xo \neq 0$.  Without loss of generality, let $\xo = e_1$. Let
\begin{align*}
L= \begin{pmatrix}\alpha & 0 \\ 0 & \beta \cdot I_{n-1 \times n-1} \end{pmatrix} \quad \text{and} \quad L^{1/2} = \begin{pmatrix}\alpha^{1/2} & 0 \\ 0 & \beta^{1/2} \cdot I_{n-1 \times n-1} \end{pmatrix},
\end{align*}
where $\alpha = \mathbb{E}[z^2 \1_{|z| \leq 3}] \approx 0.9707$ and $\beta = \PP[|z| \leq 3] \approx 0.9973$, where $z \sim \mathcal{N}(0, 1)$. 
Let $\zeta_i = \inddefwithnorm L^{-1/2} a_i$.  We have
\begin{align}
\left \| \frac{1}{m} \sum_{i=1}^m \inddefwithnorm \cdot \aiai - I_{n \times n}\right \|  &\leq \left \| L^{1/2}  \Biggl( \frac{1}{m} \sum_{i=1}^m \zetaizetai - I \Biggr) L^{1/2}\right \| + \|L  - I\|  \notag \\
&\leq \left \|    \frac{1}{m} \sum_{i=1}^m \zetaizetai - I  \right \| \|L^{1/2}\|^2 + \| L - I\|  \notag \\
& \leq \left \|   \frac{1}{m} \sum_{i=1}^m \zetaizetai - I  \right \| + 0.03. \label{eq:indaiai}
\end{align}
where the third inequality follows from $\|L^{1/2}\| \leq 1$ and $\|L-I \| = 1-\alpha < 0.03$.  

Observe that $\mathbb{E} [\zetaizetai] = I_{n \times n}$ and let $A$ be the $m \times n$ matrix whose $i$th row is given by $\zeta_i$.   By Theorem 5.39 in \cite{V2012}, there exist positive $C, \gamma$ such that for all $t \geq 0$, with probability at least $1 - 2e^{-\gamma t^2}$, 
\begin{align}
1 - C\sqrt{\frac{n}{m}} - \frac{t}{\sqrt{m}} \leq \sigmin \Bigl(\frac{A}{\sqrt{m}} \Bigr)  \leq  \sigmax \Bigl(\frac{A}{\sqrt{m}} \Bigr)  \leq 1 + C\sqrt{\frac{n}{m}} + \frac{t}{\sqrt{m}}. \notag
 \end{align}
Letting $\delta = C \sqrt{\frac{n}{m}} + \frac{t}{\sqrt{m}}$, we have by Lemma 5.36 in \cite{V2012}, $\|\frac{1}{m} A^\intercal A - I \| \leq 3 \max(\delta, \delta^2)$ with probability at least $1 - 2e^{-\gamma t^2}$. Let $\eps = 0.01/3$.  If $m \geq 4 \frac{C^2}{\eps^2} n$ and $t = \frac{\eps}{2}\sqrt{m}$, we have $\delta \leq \eps$.  Thus, 
\begin{align}
\left \|  \frac{1}{m} \sum_{i=1}^m \zetaizetai - I  \right \| =  \left \|  \frac{1}{m} A^\intercal A - I  \right \| \leq  3 \eps = 0.01, \label{eq:ata-bound}
\end{align}
with probability at least $1 -2 e^{-\gamma  \eps^2 m/4}$. 
The result follows by combining  \eqref{eq:indaiai} and \eqref{eq:ata-bound}.
\end{proof}

\begin{lemma}\label{lem:isometry-allomega-withind}
 Let $a_i \sim \mathcal{N}(0, I_{n \times n})$ for $i = 1 \ldots m$ be independent.  There exist universal constants $\delta, c_0, \gamma>0$ such that if $m \geq 2 c_0 n$ then on an event of probability at least $1- 2 m e^{-\gamma m/4}$, the following holds simultaneously for all $\Omega \subset [m]$ such that $|\Omega| \leq \delta m$:
\begin{align}
\left \| \frac{1}{|\Omegac|} \sum_{i \in \Omegac} \inddefwithnorm \cdot a_i a_i^\intercal  - I_{n \times n} \right \| \leq 0.04.
\end{align}
\end{lemma}
\begin{proof}
Let $c_0, \gamma$ be given by Lemma \ref{lem:aiai-concentration-ind}.  For $\delta < 1/2$, the assumption $m \geq 2 c_0 n$ implies $|\Omegac| \geq m/2 \geq c_0 n$.  By Lemma \ref{lem:aiai-concentration-ind}, for any fixed $|\Omega|$ with $|\Omegac|  \geq c_0 n$,
$$\left \| \frac{1}{|\Omegac|} \sum_{i \in \Omegac} \inddefwithnorm \cdot a_i a_i^\intercal  - I_{n \times n} \right \| \leq 0.04$$
with probability at least $1 - 2 e^{-\gamma | \Omegac|} \geq 1 - 2 e^{-\frac{\gamma}{2}  m} $.    The number of subsets of $m$ with size $\floor{\delta m}$ is  
$$
{m \choose \floor{\delta m}} \leq \left( \frac{e m}{\delta m} \right ) ^{\delta m} = \left[ \left( \frac{e}{\delta} \right)^\delta \right]^m.
$$
As $\lim_{\delta \to 0} \left( \frac{e}{\delta} \right)^\delta = 1$, for sufficiently small $\delta$, $\left( \frac{e}{\delta} \right)^\delta \leq e^{\frac{\gamma}{4} }$. Thus, by a union bound, the desired estimate holds simultaneously for all $|\Omega| \leq \delta m$ with probability at least
$$
1 - \sum_{k=1}^{\floor{\delta m}} {m \choose k} 2 e^{-\frac{\gamma}{2} m}  \geq 1 - \floor{\delta m} {m \choose \floor{\delta m}}  2 e^{-\frac{\gamma}{2}  m} \geq 1 - 2 m e^{-\frac{\gamma}{4} m}.
$$
\end{proof}

The next two lemmas establish probabilistic concentration estimates on the quantity $\sum_{i \in \Omega} |\aih|$ $=\| A_\Omega h \|_1$ for arbitrarily  selected small subsets $\Omega$, where $A$ is a matrix with i.i.d. standard Gaussian entries. 

\begin{lemma} \label{lem:aiho-concentration-k}
Let $A$ be a $k \times n$ matrix with i.i.d. $\mathcal{N}(0,1)$ entries and fix $C>0$.  If $k \geq n$, then on an event of probability at least $1 - 2 e^{- C^2 k /2}$, the following holds simultaneously for all $h \in \R^n$:
$$
\| A h\|_1 \leq (C+2) k \| h\|_2.
$$
\end{lemma}
\begin{proof}
 By Corollary 5.25 in \cite{V2012}, the spectral norm of $A$ is bounded by 
\begin{align*}
\| A \| \leq \sqrt{k} + \sqrt{n} + t \quad \text{ with probability at least } 1 - 2 e^{-t^2/2}.
\end{align*}
By $k \geq n$, and by taking $t  = C\sqrt{k}$, we have $\| A \| \leq (C+2) \sqrt{k}$ on an event with probability at least $1 - 2 e^{-C^2 k / 2}$. 
On this event, we conclude
\begin{align*}
\|A h\|_1 \leq \sqrt{k} \|Ah\|_2 \leq \sqrt{k} \|A\| \|h\|_2 \leq (C+2) k \|h\|_2.
\end{align*}
\end{proof}

\begin{lemma} \label{lem:aiho-concentration-uniformomega}
Fix $\delta >0$.  Let $A$ be an $m \times n$ matrix with i.i.d. $\mathcal{N}(0,1)$ entries.  If $m \geq \frac{1}{\delta} n$, then on an event of probability at least $1 - 2 e^{-\delta m} $, the following holds simultaneously for all $\Omega \subset [m]$ with $|\Omega| \leq \delta m$ and for all $h \in \R^n$: 
$$
\| A_{\Omega} h\|_1 \leq \Bigl(\sqrt{4+2\log(1/\delta)}+2 \Bigr) \delta m \| h\|_2.
$$
\begin{proof}
Without loss of generality, it suffices to only consider $\Omega$ such that  $|\Omega| = \floor{\delta m} \geq n$. Let $C = \sqrt{4 + 2\log(1/\delta)}$. By Lemma \ref{lem:aiho-concentration-k}, with probability at least $1 - 2 e^{-C^2 \floor{\delta m}/2}$, we have $\| A_\Omega h\|_1 \leq (C+2)\delta m \|h\|_2.$  The result follows by a union bound and by noting that
\begin{align*}
1 - {m \choose \floor{\delta m}}  2 e^{-C^2 \delta m/2} & \geq  1 - \left( \frac{e m}{\delta m} \right ) ^{\delta m} 2 e^{-C^2 \delta m/2}\\
& \geq  1 - e^{\delta m(1 + \log(1/\delta))}2 e^{-C^2 \delta m/2}\\
& \geq  1 - 2 e^{ \delta m( 1 + \log(1/\delta) -\frac{1}{2}C^2)}\\
& \geq  1 - 2 e^{ -\delta m}.
\end{align*}
\end{proof}

\end{lemma}

The final three lemmas establish probabilistic concentration lower-bounds on the average of   $\inddefwithnorm \cdot |\aixo | |\aixone |$ over arbitrarily  selected subsets $\Omegac$.

\begin{lemma}  \label{lemma:lower-bound-ind}
Fix $\xo \in \R^n$.  Let $a_i \sim \mathcal{N}(0, I_{n \times n})$ be independent for $i = 1 \ldots m$.  There exist constants $c_0, \gamma$ such that if $m \geq c_0 n$, then with probability at least $1-5 e^{-\gamma m}$,
\begin{align}
\frac{1}{m} \sum_{i=1}^m \inddefwithnorm \cdot |\aixo | |\aixone | \geq 0.55 \| x_0\|_2 \|x_1\|_2 \quad \text{ for all } x_1 \in \R^n. \label{eq:aiaixx-lower-bound}
\end{align}
\end{lemma}
\begin{proof} 
Without loss of generality, we take $\|x_0\|_2 = \|x_1\|_2 = 1$. 
Define $\Axo: \R^{n} \to \R^m$ by $\Axo(x_1) = (\inddef \cdot \aixo \aixone)_{i=1}^m$.  Note that $\Axo$ is linear and the left hand side of \eqref{eq:aiaixx-lower-bound} is $\frac{1}{m} \| \Axo (x_1) \|_1$. 
 
First, we show that for any fixed unit vector $x_1$, 
$\frac{1}{m} \| \Axo (x_1)\|_1 \geq 0.59$
with probability at least $1 - 3 e^{-\gamma m }$.  To show this, observe that $\frac{1}{m}\|\Axo(x_1)\|_1 = \frac{1}{m} \sum_{i=1}^m \xi_i$, where $\xi_i = \inddef \cdot | \langle a_i, x_0\rangle| |\langle a_i, x_1\rangle| $ is a sub-gaussian random variable.  Let $K$ be the sub-gaussian norm of $\xi_i - \mathbb{E}\xi_i$.  By Lemma~\ref{lemma:expectation-aiaixx}, $\mathbb{E} \xi_i \geq 0.597$.     Thus, by a Hoeffding-type inequality  (Proposition 5.10  in \cite{V2012}), for some $c>0$,
$$
\PP\Bigl( \frac{1}{m} \| \Axo(x_1)\|_1 \geq 0.59 \Bigr) \geq 1- e \cdot e^{-cm/K^2} \geq 1 - 3 e^{-\gamma m},
$$
for some $\gamma >0$.

Second, we show that the inequality in \eqref{eq:aiaixx-lower-bound} holds simultaneously for all $x_1\in\R^n$.  Let $\mathcal{N}_\eps$ be an $\eps$-net of $S^{n-1}$, where $\eps$ will be specified below.  By Lemma 5.2 in \cite{V2012}, we may take $| \mathcal{N}_\eps| \leq \bigl(1 + \frac{2}{\eps} \bigr)^n$.   For any $x_1$, there exists an $\xtilde_1 \in \mathcal{N}_\eps$ such that 
$
\| x_1- \xtilde_1\|_2
\leq  \eps.
$
Because $\Axo(x_1) = \Axo( \xtilde_1) + \Axo(x_1  - \xtilde_1)$, we have
$$
\frac{1}{m}\|\Axo( x_1)\|_1 \geq \frac{1}{m} \| \Axo( \xtilde_1) \|_1 - \frac{1}{m}\|\Axo( x_1 - \xtilde_1)\|_1.
$$
On the intersection of the events 
\begin{align*}
E_1 &= \left \{\frac{1}{m} \|\Axo(x_1) \|_1 \geq 0.59, \text{ for all }  x_1 \in \mathcal{N}_\eps \right \},\\
E_2 &= \left \{\frac{1}{m}\| \Axo(h) \|_1 \leq 9 \|h\|_2, \text{ for all } h \in \R^{n} \right \},
\end{align*}
we have
\begin{align*}
\frac{1}{m}\|\Axo( x_1)\|_1 \geq 0.59 - 9 \| x_1 - \xtilde_1\|_2 \geq (0.59 - 9 \eps)  = 0.55,
\end{align*}
by choosing $9 \eps = 0.04$.  
It remains to estimate the probability of $E_1\cap E_2$.   We have by a union bound that $$\PP(E_1) \geq 1 - | \mathcal{N}_\eps | \cdot 3 e^{-\gamma m} \geq 1 - 3 \left ( 1 + \frac{2}{\eps}\right )^{n} e^{- \gamma m}.$$  If $m \geq c_0 n$ for a sufficiently large $c_0$,  then $\PP(E_1) \geq 1 - 3 e ^{-\gamma m/2}$.  
To bound $\PP(E_2)$, note that
\begin{align*}
\frac{1}{m} \| \Axo(x_1) \|_1  = \frac{1}{m} \sum_{i=1}^m \inddef \cdot |\aixo| |\aixone| \leq \frac{3}{m} \sum_{i=1}^m | \aixone|. 
\end{align*}
By Lemma \ref{lem:aiho-concentration-k} with $C=1$ and by $m \geq n$, there is an event $E_3$, with $\PP(E_3) \geq 1 - 2 e^{-m/2}$, on which $\frac{1}{m} \sum_{i=1}^m | \aixone| \leq 3 \| x_1 \|_2$ for all $x_1 \in \R^n$.  Notice that $E_2 \subset E_3$, and hence $\PP(E_2) \geq 1 - 2 e^{-m/2}$.  
Thus, $\PP(E_1 \cap E_2) \geq 1 - 5 e^{-\gamma m}$ for some constant $\gamma$.

\end{proof}

\begin{lemma}  \label{lemma:lower-bound-ind-uniformomega}
Fix $\xo \in \R^n$.  Let $a_i \sim \mathcal{N}(0, I_{n \times n})$ be independent for $i = 1 \ldots m$.  There exist positive constants $c_0, \gamma, \delta$ such that if $m \geq 2 c_0 n$, then with probability at least $1-5 m e^{-\gamma m/4}$, the following holds simultaneously for all $\Omega \subset [m]$ such that $|\Omega| \leq \delta m$:
\begin{align}
\frac{1}{|\Omegac|} \sum_{i\in\Omegac} \inddefwithnorm \cdot |\aixo | |\aixone | \geq 0.55 \| x_0\|_2 \|x_1\|_2 \quad \text{ for all } x_1 \in \R^n. \label{eq:aiaixx-lower-bound-uniformomega}
\end{align}
\end{lemma}
\begin{proof}
Let $c_0, \gamma$ be the constants from Lemma \ref{lemma:lower-bound-ind}.   For any fixed $\Omega$ such that $|\Omega| \leq \delta m$, and for any $\delta < 1/2$, we have \eqref{eq:aiaixx-lower-bound-uniformomega} with probability at least $1 - 5 e^{-\gamma |\Omegac|} \geq 1 - 5 e^{-\gamma m / 2}$.  By a union bound and by selecting $\delta$ small enough such that $(e/\delta)^\delta < e^{\gamma/4}$, the result holds with probability at least
\begin{align*}
1 - \floor{\delta m} {m \choose \floor{\delta m}}  5 e^{-\gamma m/2} & \geq  1 - \delta m \left( \frac{e m}{\delta m} \right ) ^{\delta m} 5  e^{-\gamma m/2} \geq 1 - 5 m e^{-\gamma m/4}.
\end{align*}
\end{proof}

\begin{lemma}  \label{lemma:expectation-aiaixx}
Fix $x_0, x_1 \in \mathbb{R}^n$.  If $a \sim \mathcal{N}(0, I_{n \times n})$, then 
$\mathbb{E} \Bigl[ \inddef \cdot |\langle a, x_0\rangle| | \langle a, x_1 \rangle |  \Bigr]  \geq 0.597 \|x_0\|_2 \|x_1\|_2.$
\end{lemma}
\begin{proof}
Without loss of generality, take $\|x_0\|_2=\|x_1\|_2 = 1$.  Further, without loss of generality, take $x_0 = e_1$ and $x_1 = \cos \theta \ e_1 + \sin \theta \ e_2$, where $\theta$ is the angle between $\xo$ and $x_1$.  The desired expected value is 
\begin{align*}
\mathbb{E} \Bigl[ \mathbbm{1}_{|a_1| \leq 3} \cdot  \Big |a_1 (a_1 \cos \theta + a_2 \sin \theta) \Big | \Bigr]  
&\geq \mathbb{E} \Bigl[ \mathbbm{1}_{\sqrt{a_1^2 +a_2^2}  \leq 3} \cdot \Big|a_1 (a_1 \cos \theta + a_2 \sin \theta) \Big | \Bigr]  \\
&= \frac{1}{2 \pi} \int_0^3 r^3 e^{-r^2/2} dr \int_0^{2\pi} |\cos \phi \cos(\theta - \phi)| d\phi \\
&= \frac{1}{2 \pi} (2 - 11 e^{-9/2}) \int_0^{2\pi} |\cos \phi \cos(\theta - \phi)| d\phi \\
&=\frac{2 - 11 e^{-9/2}} {2 }\frac{1}{2\pi} \int_0^{2\pi}|\cos \theta + \cos(2 \phi - \theta)| d\phi\\
&=\frac{2 - 11 e^{-9/2}} {2 }\frac{1}{\pi} \int_0^\pi |\cos \theta + \cos \tilde{\phi}| d\tilde{\phi}\\
&=\frac{2 - 11 e^{-9/2}} {\pi }  \left(|\sin \theta| + \sin^{-1} (\cos \theta) \cos \theta \right)\\
&\geq \frac{2 - 11 e^{-9/2}} {\pi } \geq 0.597,
\end{align*}
where the third equality is because $2 \cos \phi \cos(\theta - \phi) = \cos \theta + \cos(2 \phi - \theta)$. 
\end{proof}

%

\subsection*{Acknowledgements}
PH acknowledges funding by the grant NSF DMS-1464525. 
\bibliographystyle{plain}
\bibliography{refs}

\end{document}